\documentclass[journal]{IEEEtran}
\usepackage[ruled,vlined,linesnumbered]{algorithm2e}
\usepackage{cite}
\usepackage{amsmath}
\usepackage{mathtools}
\usepackage{amssymb}
\usepackage{amsthm}
\usepackage{booktabs}
\usepackage{multirow}
\usepackage{graphicx}
\usepackage{subcaption}
\mathtoolsset{showonlyrefs}
\graphicspath{{figures/}{./figures/}}

\theoremstyle{plain}
\newtheorem{thm}{Theorem}

\newtheorem{prop}[thm]{Proposition}
\theoremstyle{definition}

\theoremstyle{remark}

\begin{document}

\title{RPG-VST: Robust Poisson–Gaussian Variance Stabilization for Blind RAW Denoising}

\author{Wenbin Wang, Xiaotong Luo, Yuan Gao*, Wenjun Zeng%
\thanks{W. Wang, Y. Gao and W. Zeng are with the OmniVision-IDT Joint Laboratory for Intelligent Image Sensing,
Ningbo Key Laboratory of Spatial Intelligence and Digital Derivative,
Ningbo Institute of Digital Twin, Eastern Institute of Technology, Ningbo;
Zhejiang Key Laboratory of Industrial Intelligence and Digital Twin.
X. Luo is with The Hong Kong Polytechnic University.
*\,Corresponding author: Y. Gao (e-mail: ygao@idt.eitech.edu.cn).}}

\markboth{IEEE Signal Processing Letters}{Robust PG Variance Stabilization for Blind RAW Denoising}
\maketitle

\begin{abstract}
Variance stabilization with the generalized Anscombe transform (GAT) enables frozen Gaussian denoisers to process Poisson--Gaussian (PG) RAW noise, but its reliability depends on fitted shot/read-noise parameters.
In blind single-image deployment, these parameters are estimated from low-texture RAW statistics that are often corrupted by residual texture, clipping, defective pixels, and read-noise floors.
Such contamination yields heavy-tailed log-variance residuals, making ordinary least-squares PG calibration brittle and causing severe tail failures despite favorable average PSNR.
We propose RPG-VST, a robust no-reference variance-stabilization framework for blind RAW denoising.
RPG-VST estimates PG parameters separately for each color filter array (CFA) plane using a Student-$t$ log-variance objective with robust tile statistics and physical constraints.
It then estimates the stabilized-domain noise level $\sigma_z$ from tile variance ratios and uses it as a reliability signal.
For each image, RPG-VST selects the robust fit or the conventional OLS fit according to which produces $\sigma_z$ closer to unit variance, requiring no clean reference or learned threshold.
On SID Sony SID$50$, SIDD, and ELD with frozen SwinIR and Restormer denoisers, RPG-VST improves mean PSNR in all six dataset--backbone settings.
It reduces severe tails, defined as cases whose PSNR gain over Direct is below $-1$\,dB, in four settings and leaves them unchanged in the other two.
On SIDD, it yields $+1.83/+1.92$\,dB and reduces severe tails from $44/36$ to $7/4$.
Ablations show that the $\sigma_z$ gate prevents regressions of ungated robust fitting on read-noise-dominated ELD captures.
\end{abstract}

\begin{IEEEkeywords}
Blind RAW denoising, Poisson--Gaussian noise, variance stabilization, generalized Anscombe transform, robust noise estimation, no-reference reliability.
\end{IEEEkeywords}

\begin{figure*}[!t]
\centering
\begin{subfigure}[t]{0.32\textwidth}
\centering
\includegraphics[width=\textwidth]{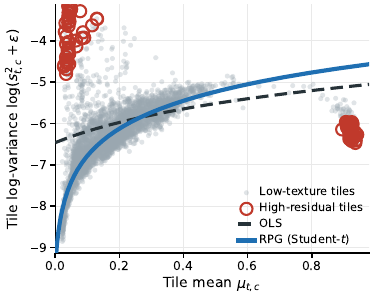}
\caption{Robust log-variance PG fit.}
\label{fig:motivation:scatter}
\end{subfigure}
\hfill
\begin{subfigure}[t]{0.32\textwidth}
\centering
\includegraphics[width=\textwidth]{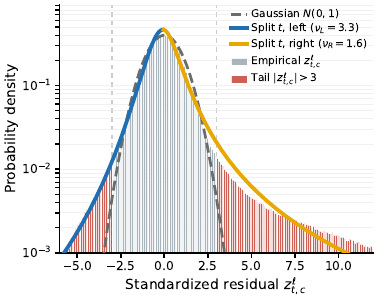}
\caption{Standardized-residual density.}
\label{fig:motivation:hist}
\end{subfigure}
\hfill
\begin{subfigure}[t]{0.32\textwidth}
\centering
\includegraphics[width=\textwidth]{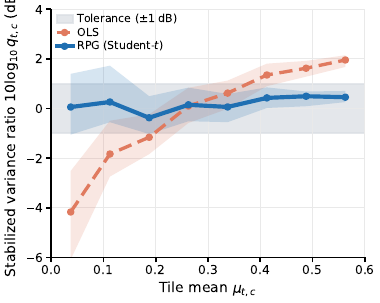}
\caption{Stabilized-domain variance ratio.}
\label{fig:motivation:tail}
\end{subfigure}
\caption{Heavy-tailed PG-fitting residuals motivate robust variance stabilization.
(\subref{fig:motivation:scatter})~Tile-level log-variance in a representative RAW crop.
High-residual tiles ($|z^{\ell}_{t,c}|>3$, red) deviate from the inlier mean--variance trend and bias the OLS fit, whereas the Student-$t$ fit follows the inliers.
(\subref{fig:motivation:hist})~Probability density of standardized residuals $z^\ell_{t,c}$ on a log scale.
The Gaussian reference underestimates both tails beyond $|z^\ell_{t,c}|>3$, while a continuous split Student-$t$ model with $\nu_L=3.3$ and $\nu_R=1.6$ matches the empirical tails, yielding $R_-(3)=9.4\times$ and $R_+(3)=28.9\times$.
(\subref{fig:motivation:tail})~Stabilized-domain variance ratio $q_{t,c}$ in dB on the same crop, shown by binned medians with $25$--$75\%$ bands.
The OLS ratio drifts outside the $\pm1$\,dB tolerance band with intensity, whereas the Student-$t$ fit remains close to the unit-variance target.}
\label{fig:motivation}
\end{figure*}

\section{Introduction}
\label{sec:intro}

\IEEEPARstart{M}{odern} image restoration increasingly relies on denoising priors trained under additive white Gaussian noise (AWGN)~\cite{buades2005review,fan2019brief,tian2020deep,zhao2026methods}.
CNN denoisers and plug-and-play priors such as DnCNN~\cite{zhang2017beyond}, FFDNet~\cite{zhang2018ffdnet}, and DRUNet~\cite{zhang2021plug} are widely used for Gaussian denoising and inverse restoration, while Transformer-based backbones such as SwinIR~\cite{liang2021swinir} and Restormer~\cite{zamir2022restormer} further improve high-resolution restoration.
These models are attractive as frozen denoisers because they can be reused without task-specific retraining.
RAW sensor measurements, however, are not governed by signal-independent AWGN~\cite{foi2008practical}.
Photon shot noise and electronic read noise produce signal-dependent Poisson--Gaussian (PG) observations, so applying a Gaussian denoiser to RAW data requires either RAW-specific training or a physical noise adaptation~\cite{brooks2019unprocessing}.

To address the gap between RAW noise and the AWGN assumption, existing real RAW denoising studies have mainly developed along data-driven denoising and camera noise modeling~\cite{foi2008practical,abdelhamed2019noise,wei2020physics,cao2023physics}.
On the data side, paired real RAW datasets such as SID and SIDD have promoted the development of low-light and smartphone RAW denoising~\cite{chen2018learning,abdelhamed2018high}.
On the modeling side, Unprocessing, CBDNet, Noise Flow, ELD, and ISO-dependent sensor noise models improve noise realism from complementary perspectives, including inverse camera pipelines, realistic noise synthesis, conditional noise distribution modeling, and low-light sensor physics~\cite{brooks2019unprocessing,guo2019cbdnet,abdelhamed2019noise,wei2020physics,cao2023physics}.
Related RAW restoration methods also operate directly on packed CFA measurements~\cite{chen2018learning,wei2020physics}.
These studies show that real camera noise is far more complex than AWGN or a single global noise level.
However, they typically rely on paired data, sensor calibration, camera-dependent synthesis, or learned noise models, and therefore do not directly address blind RAW denoising when the goal is to reuse a frozen Gaussian denoiser without retraining.
In this setting, a natural alternative is not to learn another RAW denoiser, but to use a physical transform that maps signal-dependent RAW noise to an approximately Gaussian and approximately constant-variance domain, where existing Gaussian denoisers remain applicable.
Noise2VST~\cite{herbreteau2025noise2vst} learns a model-free VST from each noisy image using a frozen Gaussian denoiser, while FBI-Denoiser~\cite{byun2021fbi} couples a learned PG-estimation network with a blind-spot denoiser.
In contrast, RPG-VST retains the analytic GAT and focuses on robust CFA-wise PG calibration with no-reference selection between robust and OLS estimates.

Following this variance-stabilization route, the Anscombe transform and its generalized form for PG noise map signal-dependent observations to an approximately constant-variance domain, where a Gaussian denoiser can be applied before an inverse transform returns the result to the original domain~\cite{anscombe1948transformation,foi2008practical,makitalo2013optimal}.
GAT therefore provides a physically interpretable adaptation path for processing RAW PG noise with frozen Gaussian denoisers.
The reliability of this path, however, depends not only on the denoiser itself, but also on the shot- and read-noise parameters used by GAT.
Prior work on PG fitting and mismatch analysis has shown that parameter errors directly break the unit-variance target of variance stabilization~\cite{foi2008practical,makitalo2014mismatch}.
This issue becomes especially critical in blind single-image deployment, where the PG parameters must be estimated from low-texture RAW tiles in the current image, while these tiles are often contaminated by residual texture, clipping, defective pixels, row or column artifacts, and read-noise floors~\cite{makitalo2014mismatch,abdelhamed2019noise,wei2020physics,cao2023physics}.
The resulting log-variance residuals are heavy-tailed, making an ordinary least-squares (OLS) PG fit vulnerable to a small number of unreliable tiles.
Heavy-tailed scale-mixture models limit the influence of contaminated statistics~\cite{huang2017mixed}; our Student-$t$ objective applies this principle specifically to PG-fitting residuals, not to the RAW noise distribution itself.
A biased PG fit breaks the unit-variance target of GAT and leaves residual heteroscedastic noise for the frozen denoiser.
Consequently, a GAT wrapper may improve average PSNR while still producing severe negative tails on individual images.

Fig.~\ref{fig:motivation} further illustrates this failure mode.
In a representative RAW region, most low-texture tiles follow a coherent mean--variance trend, whereas high-residual tiles pull the OLS fit away from the inlier trend.
The standardized log-variance residuals have much heavier tails than a Gaussian model, with tail probabilities beyond $|z|>3$ reaching $9.4\times$ and $28.9\times$ the Gaussian rate.
After variance stabilization, the OLS variance ratio varies with intensity and moves outside the $\pm1$\,dB tolerance band, while the Student-$t$ log-variance fit remains close to unit variance.
These observations indicate that the main risk in blind GAT is not merely whether GAT is used, but how the PG parameters required by GAT can be robustly estimated and validated without a reference image.

Based on this analysis, we propose RPG-VST, a robust PG-fitted variance-stabilization framework for blind RAW denoising with frozen Gaussian denoisers.
RPG-VST first estimates per-color-filter-array (CFA) PG parameters from robust tile statistics using a Student-$t$ loss on log-variance residuals, and then estimates the stabilized-domain noise level $\sigma_z$ from tile variance ratios.
The estimated $\sigma_z$ is used as a no-reference reliability signal to choose between the robust fit and the conventional OLS fit.
Since the OLS branch corresponds to Vanilla-GAT, the gate provides a conservative fallback when the robust fit is unreliable.
This work makes three contributions.
\begin{itemize}
    \item First, it analyzes how PG-parameter mismatch propagates through GAT and induces an intensity-dependent variance ratio in the stabilized domain.
    \item Second, it introduces a heavy-tailed log-variance PG estimator with robust tile statistics and physical constraints.
    \item Third, it introduces a no-reference $\sigma_z$ gate for robust or OLS selection in blind deployment. Experiments on SID Sony SID$50$~\cite{chen2018learning}, SIDD RAW blocks~\cite{abdelhamed2018high}, and ELD~\cite{wei2020physics} with SwinIR and Restormer show that RPG-VST improves mean PSNR in all six dataset--backbone settings and reduces severe-tail failures.
\end{itemize}

\section{Robust PG-Fitted Variance Stabilization}
\label{sec:method}

\subsection{PG RAW Model and Variance-Stabilized Denoising}

Let $y_i$ denote a normalized RAW pixel, and let $c(i)\in\mathcal C$
indicate its packed CFA plane, with
$\mathcal C=\{R,G_1,G_2,B\}$ in a Bayer RAW image.
We adopt the Poisson--Gaussian (PG) RAW noise model
\begin{equation}
y_i=x_i+n_i,\qquad
n_i\,|\,x_i\sim
\mathcal N\!\left(0,a_{c(i)}x_i+b_{c(i)}\right),
\label{eq:pg_model_rpg}
\end{equation}
where $a_c>0$ and $b_c\ge0$ denote the shot-noise and read-noise
parameters of CFA plane $c$.
For $\theta_c=(a_c,b_c)$, the generalized Anscombe transform (GAT) is
\begin{equation}
T_{\theta_c}(u)
=
\frac{2}{a_c}
\sqrt{a_cu+\frac{3}{8}a_c^2+b_c},
\label{eq:gat_rpg}
\end{equation}
with its exact unbiased or closed-form inverse denoted by
$T_{\theta_c}^{-1}$~\cite{foi2008practical,makitalo2013optimal}.
For a RAW image, $T_\theta$ is applied plane-wise according to $c(i)$.

Given a frozen Gaussian denoiser $D$, the VST denoising pipeline is
\begin{equation}
\hat x_{\theta,\sigma_z}
=
T_{\theta}^{-1}
\!\left(
D\!\left(T_{\theta}(y);\sigma_z\right)
\right),
\label{eq:rpg_vst_output}
\end{equation}
where $\sigma_z$ is the noise level supplied to the denoiser in the
stabilized domain.
Thus, blind RAW denoising with a frozen Gaussian denoiser reduces to
estimating the PG parameters $\theta$ and the stabilized-domain noise
scale $\sigma_z$ from the input image itself.

The effect of PG-parameter mismatch explains why this estimation step is
critical.
\begin{prop}[Effect of PG-parameter mismatch]
\label{prop:pg_mismatch}
Consider one CFA plane.
Suppose that the true local noise variance is $a^\star x_i+b^\star$,
whereas GAT is applied with parameters $(\bar a,\bar b)$.
To first order in the noise amplitude,
\begin{equation}
\operatorname{Var}
\!\left[
T_{\bar a,\bar b}(y_i)\mid x_i
\right]
\approx
\frac{a^\star x_i+b^\star}
{\bar a x_i+\bar b+\frac{3}{8}\bar a^2}.
\label{eq:mismatch_variance_ratio}
\end{equation}
\end{prop}

\begin{proof}
A first-order Taylor expansion of $T_{\bar a,\bar b}$ around $x_i$ gives
$T_{\bar a,\bar b}'(x_i)
=(\bar a x_i+\bar b+\tfrac38\bar a^2)^{-1/2}$.
Taking the conditional variance of
$T_{\bar a,\bar b}'(x_i)n_i$ gives
\eqref{eq:mismatch_variance_ratio}.
\end{proof}

Proposition~\ref{prop:pg_mismatch} shows that a mismatched PG fit induces
an intensity-dependent residual variance in the GAT domain, as illustrated
in Fig.~\ref{fig:motivation:tail}.
Tail failures should therefore be controlled at the PG-fitting stage.
The relevant quantity is the variance ratio implied by the fitted model,
which motivates the log-variance formulation below.

\subsection{Heavy-Tailed Log-Variance PG Fitting}

For each CFA plane $c$, we mask saturated, near-black, and defective pixels, partition the crop into $16\times16$ tiles, and form $\mathcal F_c$ from the flattest $55\%$ according to fixed gradient and high-pass scores.
These fitting settings are fixed a priori and shared across datasets and cameras.
For a tile $t\in\mathcal F_c$ with valid pixels $\Omega_{t,c}$, we compute
a robust local mean
$\mu_{t,c}=\operatorname{med}_{i\in\Omega_{t,c}}y_i$
and a robust local variance
\begin{equation}
s_{t,c}^2
=
\operatorname{WVar}_{\tau}
\{y_i-\hat p_{t,c}(i)\}_{i\in\Omega_{t,c}},
\label{eq:robust_tile_variance}
\end{equation}
where $\hat p_{t,c}$ is a Huber-fitted affine shading plane and
$\operatorname{WVar}_{\tau}$ denotes a Winsorized variance with $\tau=0.05$.
This tile statistic suppresses residual texture and isolated outliers while
preserving the local mean--variance relation needed for PG fitting.

We parameterize $a_c=\exp(\alpha_c)$ and $b_c=\exp(\beta_c)$, and write
$\vartheta_c=(\alpha_c,\beta_c)$.
The PG variance predicted for tile $t$ is
\begin{equation}
v_{t,c}(\vartheta_c)
=
\exp(\alpha_c)\mu_{t,c}+\exp(\beta_c).
\label{eq:tile_pg_variance}
\end{equation}
Instead of fitting $s_{t,c}^2$ to $v_{t,c}$ by least squares in the
variance domain, we use the log-variance residual
\begin{equation}
\ell_{t,c}(\vartheta_c)
=
\log(s_{t,c}^2+\varepsilon)
-
\log(v_{t,c}(\vartheta_c)+\varepsilon).
\label{eq:log_variance_residual}
\end{equation}
This residual measures multiplicative variance mismatch, which is the
quantity that directly appears after variance stabilization.

The robust PG estimate is obtained by
\begin{equation}
\hat\vartheta_c
=
\arg\min_{\vartheta_c\in\Theta_c}
\sum_{t\in\mathcal F_c}
w_{t,c}\,
\rho_\nu
\!\left(
\frac{\ell_{t,c}(\vartheta_c)}
{\hat s_{\ell,c}}
\right)
+
\lambda
\left\|
\vartheta_c-\vartheta_{0,c}
\right\|_{\Sigma_c^{-1}}^2,
\label{eq:rpg_objective}
\end{equation}
where
\begin{equation}
\rho_\nu(u)
=
\frac{\nu+1}{2}
\log\!\left(1+\frac{u^2}{\nu}\right)
\label{eq:student_t_loss}
\end{equation}
is the Student-$t$ negative log-likelihood with $\nu=3$, up to an additive constant.
Here, $w_{t,c}$ is a fixed tile weight determined by the number of valid
pixels and the flatness score, and $\hat s_{\ell,c}$ is a robust scale
estimate of the initial log-variance residuals.
The feasible set $\Theta_c$ imposes the physical bounds
\begin{equation}
\exp(\alpha_c)\in[10^{-8},2\!\times\!10^{-2}], \qquad \exp(\beta_c)\in[10^{-10},10^{-2}].
\label{eq:physical_bounds}
\end{equation}
The quadratic term in \eqref{eq:rpg_objective} is optional.
When camera- or ISO-level calibration priors are unavailable, we set
$\lambda=0$.
Otherwise, $(\vartheta_{0,c},\Sigma_c)$ provides a weak sensor prior.

We initialize \eqref{eq:rpg_objective} with a Theil--Sen-type robust
linear fit in the mean--variance domain, project the result onto
$\Theta_c$, and refine it using a few iteratively reweighted least-squares
steps.
The implicit IRLS tile weight is
\begin{equation}
\omega_{t,c}^{\mathrm{IRLS}}
=
w_{t,c}
\frac{\nu+1}
{\nu+
\left(
\ell_{t,c}(\vartheta_c)/\hat s_{\ell,c}
\right)^2}.
\label{eq:irls_weight}
\end{equation}
Large log-variance residuals therefore receive smaller weights, preventing
a few contaminated tiles from dominating the PG fit.
We denote the minimizer of \eqref{eq:rpg_objective} on a tile set
$\mathcal A_c$ by $\operatorname{RPGFit}(\mathcal A_c)$.

\subsection{Scale-Adaptive Robust GAT}

After fitting PG parameters, the stabilized noise level should be close to
unit variance.
In real RAW crops, residual sensor mismatch can still leave a global scale
error even when the intensity-dependent part is well calibrated.
For any fitted parameter set
$\theta=\{(a_c,b_c)\}_{c\in\mathcal C}$, we estimate this scale from the
tile-level stabilized variance ratio
\begin{equation}
q_{t,c}(\theta)
=
\frac{s_{t,c}^2+\varepsilon}
{a_c\mu_{t,c}+b_c+\frac{3}{8}a_c^2+\varepsilon}.
\label{eq:stabilized_ratio}
\end{equation}
The corresponding stabilized-domain noise level is
\begin{equation}
\hat\sigma_z(\theta)
=
\operatorname{clip}
\left(
\sqrt{
\operatorname{med}_{c\in\mathcal C,\ t\in\mathcal F_c}
q_{t,c}(\theta)
},
\sigma_{\min},
\sigma_{\max}
\right).
\label{eq:sigma_z}
\end{equation}
The robust branch uses
$\hat\theta=\{(\exp(\hat\alpha_c),\exp(\hat\beta_c))\}_{c\in\mathcal C}$
and $\hat\sigma_z=\hat\sigma_z(\hat\theta)$, producing
\begin{equation}
\hat x^{\mathrm{RPG}}
=
T_{\hat\theta}^{-1}
\!\left(
D\!\left(T_{\hat\theta}(y);\hat\sigma_z\right)
\right).
\label{eq:initial_rpg_vst}
\end{equation}

\subsection{No-Reference $\sigma_z$-Gated Robust/OLS Selection}
\label{sec:sigmaz_gate}

Robust fitting improves the PG estimate when the low-texture statistics
contain sparse contamination, but it can still be unreliable when the
statistics themselves are dominated by sensor floors or deep-shadow
artifacts.
For example, in read-noise-dominated high-ISO captures, the blind
log-variance objective may drive the read-noise estimate toward a bound and
compensate with an inflated shot-noise estimate.
In such cases, the global stabilized variance scale may depart from one, and the frozen denoiser receives a mis-scaled signal.

We use this observation as a no-reference reliability test.
Besides the robust estimate $\hat\theta$, we compute the conventional OLS
PG fit $\hat\theta^{\mathrm{OLS}}$ from the same tile statistics.
Their stabilized-domain noise levels are
$\hat\sigma_z=\hat\sigma_z(\hat\theta)$ and
$\hat\sigma_z^{\mathrm{OLS}}=\hat\sigma_z(\hat\theta^{\mathrm{OLS}})$.
Because a correctly calibrated GAT should yield unit-variance noise, $\lvert\log\hat\sigma_z\rvert$ is used as a pre-denoising scale-consistency score: a smaller value indicates closer agreement between the median tile-variance scale and the nominal unit-variance target, without using a clean reference, a reconstruction-quality metric, or a learned threshold. 
We select the fit with the smaller score:
\begin{equation}
(\theta^{\mathrm{sel}},\sigma_z^{\mathrm{sel}})
=
\begin{cases}
(\hat\theta,\hat\sigma_z),
& \lvert\log\hat\sigma_z\rvert
\le
\lvert\log\hat\sigma_z^{\mathrm{OLS}}\rvert,\\[2pt]
(\hat\theta^{\mathrm{OLS}},\hat\sigma_z^{\mathrm{OLS}}),
& \text{otherwise.}
\end{cases}
\label{eq:sigmaz_gate}
\end{equation}
The final hybrid reconstruction is
\begin{equation}
\hat x^{\mathrm{hyb}}
=
T_{\theta^{\mathrm{sel}}}^{-1}
\!\left(
D\!\left(
T_{\theta^{\mathrm{sel}}}(y);
\sigma_z^{\mathrm{sel}}
\right)
\right).
\label{eq:hybrid_reconstruction}
\end{equation}
Because the OLS branch coincides with Vanilla-GAT, the rule reverts to the conventional fit whenever its scale-consistency score is smaller.
The comparison uses only fitted PG parameters and tile statistics before denoising; it evaluates neither reconstruction and therefore requires only one frozen-denoiser call.

\begin{algorithm}[t]
\caption{RPG-VST: Robust Poisson-Gaussian Variance Stabilization for Blind RAW Denoising}
\label{alg:rpgvst}
\DontPrintSemicolon
\KwIn{RAW crop $y$, frozen Gaussian denoiser $D$}
\KwOut{Estimate $\hat x^{\mathrm{hyb}}$}
Build the valid-pixel mask and packed CFA planes
$\{\Omega_c\}_{c\in\mathcal C}$\;
Extract low-texture tile sets $\{\mathcal F_c\}$ and robust statistics
$\{(\mu_{t,c},s_{t,c}^2)\}$\;
\ForEach{$c\in\mathcal C$}{
    $\hat\vartheta_c\leftarrow\operatorname{RPGFit}(\mathcal F_c)$ via
    \eqref{eq:rpg_objective}\;
    $\hat\vartheta_c^{\mathrm{OLS}}\leftarrow$ OLS PG fit on
    $\mathcal F_c$\;
}
Form $\hat\theta$ and $\hat\theta^{\mathrm{OLS}}$\;
Compute $\hat\sigma_z(\hat\theta)$ and
$\hat\sigma_z(\hat\theta^{\mathrm{OLS}})$ using \eqref{eq:sigma_z}\;
Select $(\theta^{\mathrm{sel}},\sigma_z^{\mathrm{sel}})$ using
\eqref{eq:sigmaz_gate}\;
\Return $\hat x^{\mathrm{hyb}}$ from \eqref{eq:hybrid_reconstruction}\;
\end{algorithm}

\section{Experiments}
\label{sec:exp}

We evaluate frozen SwinIR-c$50$ and Restormer-$\sigma25$ on SID$50$ ($100$ inputs)~\cite{chen2018learning}, SIDD validation RAW ($320$ blocks)~\cite{abdelhamed2018high}, and ELD ($240$ RAW pairs from four cameras and $10$ scenes at ISO levels $\{800,1600,3200\}$ and low-light factors $\{100,200\}$)~\cite{wei2020physics}.
All methods use the same normalized packed-CFA inputs and frozen weights, and PG fitting is performed independently on each noisy input without a reference.
Direct applies the frozen denoiser without GAT or PG fitting and serves only as the evaluation reference.
Vanilla-GAT uses the same GAT/inverse and denoiser pipeline as RPG-VST but replaces the Student-$t$ PG fit with OLS on the same tiles; it is exactly the gate's OLS fallback.
RPG-VST selects between the robust and Vanilla-GAT fits using \eqref{eq:sigmaz_gate}.
PSNR is computed per input against its paired clean reference in the normalized packed-RAW domain and averaged within each dataset--backbone setting.
For method $m$, a severe tail is an input $j$ satisfying $\operatorname{PSNR}_{m}(j)-\operatorname{PSNR}_{\mathrm{Direct}}(j)<-1$\,dB under the same backbone.
Table~\ref{tab:main} reports mean PSNR and severe-tail count for Vanilla-GAT and the RPG-VST hybrid.
RPG-VST improves mean PSNR in all six dataset--backbone settings and never increases the severe-tail count.
It reduces severe tails in four settings, with the largest improvements on SIDD, where PSNR increases by $+1.83/+1.92$ dB and severe tails are reduced from $44/36$ to $7/4$ for SwinIR and Restormer, respectively.
Consistent improvements are also observed on SID$50$+SwinIR, with $+0.54$ dB and $13\!\to\!5$ tails, and on ELD+SwinIR, with $+0.19$ dB and $72\!\to\!58$ tails.
RPG-VST gains $+0.18$ dB with $4$ severe tails on SID$50$+Restormer, but only $+0.04$ dB with an unchanged tail count of $109$ on ELD+Restormer, indicating that read-noise-dominated deep shadows remain challenging.
The gate can fall back to OLS when the robust fit is unreliable, but cannot resolve cases in which both candidate fits are inadequate.

\begin{table}[!t]
\renewcommand{\arraystretch}{1.05}
\caption{RPG-VST evaluation against the Vanilla-GAT baseline. Each entry reports mean PSNR (dB) / severe-tail count, where a severe tail is a case with PSNR gain over Direct below $-1$ dB.}
\label{tab:main}
\centering
\footnotesize
\resizebox{\columnwidth}{!}{%
\begin{tabular}{@{}llcc@{}}
\toprule
Dataset & Backbone & Vanilla-GAT & RPG-VST \\
\midrule
SID$50$ & SwinIR-c50 & $41.89$ / $13$ & $\mathbf{42.43}$ / $\mathbf{5}$ \\
SID$50$ & Restormer-$\sigma25$ & $38.42$ / $\mathbf{4}$ & $\mathbf{38.60}$ / $\mathbf{4}$ \\
SIDD & SwinIR-c50 & $48.39$ / $44$ & $\mathbf{50.22}$ / $\mathbf{7}$ \\
SIDD & Restormer-$\sigma25$ & $44.90$ / $36$ & $\mathbf{46.82}$ / $\mathbf{4}$ \\
ELD & SwinIR-c50 & $26.26$ / $72$ & $\mathbf{26.45}$ / $\mathbf{58}$ \\
ELD & Restormer-$\sigma25$ & $21.19$ / $\mathbf{109}$ & $\mathbf{21.23}$ / $\mathbf{109}$ \\
\bottomrule
\end{tabular}
}
\end{table}

\begin{table}[!t]
\renewcommand{\arraystretch}{1.05}
\caption{Ablation of the $\sigma_z$ gate, pooled over SwinIR-c$50$ and Restormer-$\sigma25$ (mean PSNR, dB / severe-tail count).}
\label{tab:ablation}
\centering
\footnotesize
\resizebox{\columnwidth}{!}{%
\begin{tabular}{@{}lccc@{}}
\toprule
Method & SID$50$ & SIDD & ELD \\
\midrule
Vanilla-GAT & $40.15/17$ & $46.64/80$ & $23.72/181$ \\
RPG-VST without $\sigma_z$ gate & $40.41/13$ & $48.51/\mathbf{11}$ & $23.84/200$ \\
RPG-VST with $\sigma_z$ gate & $\mathbf{40.52}/\mathbf{9}$ & $\mathbf{48.52}/\mathbf{11}$ & $\mathbf{23.84}/\mathbf{167}$ \\
\bottomrule
\end{tabular}
}
\end{table}

Table~\ref{tab:ablation} isolates the effect of the $\sigma_z$ gate.
Ungated robust fitting improves SID$50$ and SIDD but increases ELD tails from $181$ to $200$ under read-noise and deep-shadow contamination.
The gate improves SID$50$ ($40.41/13\!\to\!40.52/9$), raises SIDD PSNR ($48.51\!\to\!48.52$) at $11$ tails, and reduces ELD tails ($200\!\to\!167$) at $23.84$ dB.
Thus, the full method outperforms Vanilla-GAT in average PSNR and summed tails on all three datasets.

\section{Conclusion}
\label{sec:conclusion}

RPG-VST addresses tail failures of GAT-wrapped frozen denoisers at the PG-fitting stage.
It combines a heavy-tailed Student-$t$ log-variance estimator with a no-reference $\sigma_z$ gate that selects, for each image, between the robust PG fit and the conventional OLS fit according to the stabilized-domain noise scale.
Across SwinIR and Restormer on SID$50$, SIDD, and ELD, RPG-VST improves mean PSNR in all six dataset--backbone settings, reduces severe-tail counts in four, and leaves them unchanged in the remaining two.
The ablation further shows that the $\sigma_z$ gate is essential for avoiding regressions of the ungated robust fit on read-noise-dominated ELD captures.
Future work will study cross-CFA parameter sharing, broader frozen backbones, and finer-grained tile-level reliability gating.

\clearpage
\bibliographystyle{IEEEtran}
\bibliography{references}

\end{document}